\documentclass[a4paper,11pt]{article}

\usepackage[utf8]{inputenc}
\usepackage[backend=biber,style=alphabetic,maxbibnames=100,maxalphanames=4]{biblatex}
\usepackage[in]{fullpage}
\usepackage[dvipsnames,svgnames,x11names]{xcolor}
\usepackage[colorlinks]{hyperref}
\usepackage{enumerate}
\usepackage{nicefrac}
\usepackage{algorithm}
\usepackage{amsmath}
\usepackage{hyperref}
\usepackage{cleveref}
\usepackage{float}
\usepackage{subcaption}
\usepackage{graphicx}
\usepackage{amsfonts,amssymb,amsthm,bbm,mathtools}
\usepackage{mathtools}
\usepackage{authblk}
\usepackage{amssymb}
\usepackage{amsthm}
\usepackage{esvect}
\usepackage[normalem]{ulem}
\usepackage{hhline}
\usepackage{braket}
\usepackage{pgfplots}
\usepackage{multicol}
\usepackage[final]{microtype}
\usepackage[font=small]{caption}
\usepackage{diagbox}
\usepackage{thm-restate}

\usetikzlibrary{calc}
\usetikzlibrary{trees}
\usetikzlibrary{patterns}
\usepgfplotslibrary{fillbetween}

\pgfplotsset{compat=1.6}

\bibliography{bibliography}

\hypersetup{
    breaklinks = true,
	colorlinks = true,
	citecolor = teal,
	urlcolor = purple,
}

\newtheorem{theorem}{Theorem}
\newtheorem{lemma}[theorem]{Lemma}

\newcommand{\R}{\mathrm{R}}

\DeclareMathOperator{\Oh}{O}

\DeclareMathOperator{\be}{H}

\DeclareMathOperator{\poly}{poly}

\title{Quantum Time-Space Tradeoffs for\\ Exponential Dynamic Programming}
\author[$1$]{Susanna Caroppo}
\author[$2$]{Jevgēnijs Vihrovs}
\author[$2$]{\authorcr Dārta Zajakina}
\author[$2$]{Aleksejs Zajakins}

\affil[$1$]{Roma Tre University, Department of Engineering}
\affil[$2$]{University of Latvia, Faculty of Science and Technology}

\date{}

\begin{document}

\maketitle

\begin{abstract}
    We investigate the quantum algorithms for dynamic programming by Ambainis et al.~(SODA'19).
    While giving provable complexity speedups and applicable to a variety of NP-hard problems, these algorithms have a notable drawback: they require a large amount of Quantum Random Access Memory (QRAM), which potentially could be very challenging to implement in a physical quantum computer.
    In this work, we study how we can improve the space complexity by trading it for time, while still retaining a speedup over the classical algorithms.
    We show novel quantum time-space tradeoffs by combining different classical approaches with quantum techniques.
    For instance, we show that the \textsc{Travelling Salesman Problem} can be solved quantumly in $\widetilde\Oh(1.859^n)$ time and $\widetilde\Oh(1.315^n)$ QRAM space.
\end{abstract}

\section{Introduction}

In this work, we investigate quantum algorithms designed to address NP-hard problems more efficiently than the best known classical algorithms.
The most well-known example of such an algorithm is Grover's search algorithm \cite{Grover96}, which can be applied to explore all possible $2^n$ assignments of a SAT formula with $n$ variables.
The resulting $\widetilde \Oh(\sqrt{2^n})$ time complexity constitutes a quadratic speedup over the respective classical algorithm.

For many other NP-hard problems, however, simply applying Grover's search over all possible options doesn't result in any interesting quantum speedup.
Ambainis et al.~showed that Grover's search can be combined with the dynamic programming (DP) of Bellman, Held and Karp \cite{Bel62, HK62} to obtain quantum speedups for a multitude of NP-hard problems \cite{ABIKPV19}.
This, notably, gave an $\widetilde \Oh(1.728^n)$ time algorithm for the \textsc{Travelling Salesman Problem} and $\widetilde \Oh(1.817^n)$ time algorithms for general graph vertex ordering problems.
Subsequently, these ideas have been applied for many other problems, for example, see \cite{MIKL20,SM20,KPV22,GL24,Tan25,CDLDB25}.

Although these algorithms are very useful for obtaining sound theoretical quantum speedups, they typically consume a large amount of Quantum Random Access Memory (QRAM), which is a quantum analogue of RAM.
In a nutshell, given $N$ memory data items $D_1$, $\ldots$, $D_N$, a QRAM read operation implements the unitary mapping $\ket{i}\ket{0} \mapsto \ket{i}\ket{D_i}$.
As a consequence, the advantageous difference from classical RAM is that the memory can be accessed in superposition, which is crucial to the algorithms above.

The complexity of such ``quantum dynamic programming'' algorithms is measured by the number of gates in the QRAM circuit model, which is the standard circuit model augmented with QRAM gates.
The underlying assumption for the quantum speedups above is that an implementation of a QRAM gate is ``cheap'' and its running time is comparable to a RAM gate in classical computation.
On one hand, there are theoretical proposals for QRAM architectures requiring only $\Oh(\poly \log N)$ time for a QRAM operation \cite{GLM08}.
On the other hand, there is an active debate about whether this assumption is reasonable \cite{JR25}, and the progress on the physical realization of QRAM has been limited \cite{SJXWWZZZGCLYHDWZLGZSLWY25}.

Taking into account the various technical difficulties, it is reasonable to assume that if cheap QRAM is attainable, then the size of the available memory will be small, at least in the near term.
Another plausible scenario is that only ``semi-cheap'' QRAM physically is possible, having a running time of an operation between $\Oh(\poly \log N)$ and trivial $\Oh(N)$.
In this paper we investigate whether a quantum speedup of the exponential quantum dynamic algorithms is possible in the setting with QRAM of limited size.
More specifically, we study the time-space tradeoffs by fixing the amount of available QRAM memory and investigating the best possible running time under this restriction.
The problem of computing with limited memory is also relevant classically, and such time-space tradeoffs have been studied for classical exponential-time algorithms \cite{KP10}.

\subsection{Our Results}

The algorithms of \cite{ABIKPV19} are based on the classical dynamic programming over subsets \cite{Bel62, HK62}.
Conceptually, given a problem with an $n$ element input, these algorithms compute some function $f : 2^{[n]} \to \mathbb N$ for all subsets $S \subseteq [n]$.
The dynamic programming approach computes the values of $f(S)$ from smaller to larger size of $S$ by using the previously computed values; the answer to the problem is usually given by $f([n])$.

On a high level, all of these quantum algorithms share a common structure: in the first stage, they run a classical dynamic programming to compute the values $f(S)$ for $S$ of small size, $|S| \leq \alpha n$, storing all of them in QRAM memory.
The second stage runs a quantum search procedure based on Grover's algorithm and uses the fact that $f(S)$ for small $S$ can be retrieved from memory in only $\Oh(\poly n)$ time (using the cheap-QRAM assumption).
The optimal time complexity of such an algorithm is achieved when $\alpha$ is chosen so that the complexities of the two stages are balanced.
On the other hand, computing the optimal running time of the algorithm when $\alpha$ is fixed naturally gives a time-space tradeoff.
Of course, the actual algorithms have more than just one parameter $\alpha$ and the best running time depends on the choice of their values.

In this paper, we pursue two directions:
\begin{enumerate}
    \item \textbf{Optimization tradeoffs.} Calculate the quantum time-space tradeoffs of the existing algorithms by optimizing the relevant parameters.
    \item \textbf{Improved tradeoffs.} Find novel quantum time-space tradeoffs that are more efficient than the optimization tradeoffs.
\end{enumerate}

There are two different algorithms in \cite{ABIKPV19} for two classes of tasks: \textit{divide \& conquer} problems and \textit{permutation} problems.
We examine both of these classes separately.

\paragraph{Divide \& Conquer Problems.}
Broadly, a problem of this type satisfies the divide \& conquer recurrence
\begin{equation} \label{eq:dnq}
    f(S) = \min_{\substack{X \subset S \\ |X| = k}} g(f(S \setminus X), f(X)),
\end{equation}
for any value of $k \in [|S|-1]$, where $g$ is some function that can be computed in $\Oh(\poly n)$ time.
The answer then is given by $f([n])$.
Examples of such problems include the \textsc{Travelling Salesman Problem} and \textsc{Feedback Arc Set}.
For these problems \cite{ABIKPV19} provides an algorithm with time and space complexity $\widetilde \Oh(1.728^n)$.

First, we calculate the optimization tradeoff for this algorithm:
\begin{restatable}{theorem}{dnqopt}\label{thm:dnq-opt}
    Suppose that the maximum available amount $\mathcal S^n$ of QRAM space is given by $\mathcal S = 2^{\be(\alpha)} \leq 1.728$ for some $\alpha \in [0,1/2]$.
    Let $k$ be the unique integer such that $\frac{1}{2^{k+1}} \leq \alpha < \frac{1}{2^k}$.
    Then there is a quantum algorithm for the divide \& conquer problems with time complexity $\widetilde\Oh(\mathcal T^n)$, where
    $\mathcal T = 2^{1-\frac{2-\be(2^k\alpha)}{2^{k+1}}}$.
\end{restatable}
When $\alpha=1/2^{k+1}$, time complexity simplifies to $\mathcal T = 2^{1-\alpha}$, so the optimization tradeoff is lower bounded by $\mathcal T = 2^{1-\be^{-1}(\log_2 \mathcal S)}$ for all $\mathcal S$.

Next, we show an improved quantum time-space tradeoff:
\begin{restatable}{theorem}{dnqimp} \label{thm:dnq-imp}
    For any maximum available amount $\mathcal S^n$ of QRAM space between $\widetilde\Oh(\poly n)$ and $\widetilde \Oh(1.728^n)$, there is a quantum algorithm for the divide \& conquer problems with time complexity $\widetilde\Oh(\mathcal T^n)$, where $\mathcal T \in [2/\mathcal S^{0.268},2/\mathcal S^{0.201}]$.
\end{restatable}

The improved tradeoff is based on the well-known classical time-space tradeoff for such problems by Gurevich and Shelah \cite{GS87} with time complexity $\widetilde\Oh(4^n 2^{-s})$ and space complexity $\Oh(2^s)$, for any $s = n/2, n/4, n/8, \ldots$
Their technique combines divide \& conquer over subsets with dynamic programming for subproblems of size $s$.
Our result is achieved by combining this algorithm with the quantum dynamic programming algorithm of \cite{ABIKPV19}.
Notably, the $\widetilde \Oh(1.728^n)$ algorithm is already based on the same idea: hence, the improved tradeoff is obtained by utilizing the divide \& conquer \textit{twice}!

Figure \ref{fig:dnq-short} shows the full graph of the two time-space tradeoffs.
Their complexities coincide when $\mathcal S \in [1.438, 1.728]$, reason being that they compile to the same algorithm in that range.
The explicit formula for the time complexity for the improved tradeoff is derived in Section \ref{sec:dnq}.
Interestingly, the improved tradeoff follows a fractal pattern -- we call this property \textit{fractalization} and explain it in Lemma \ref{thm:fract}.
Note that for all \cite{ABIKPV19}-style algorithms the time complexity is always at least the space complexity because of the first precomputation stage.
Therefore, no tradeoff will go below the $\mathcal T = \mathcal S$ curve at any point.

We note, however, that Theorem \ref{thm:dnq-opt} has a significant practical advantage over Theorem \ref{thm:dnq-imp}.
The first tradeoff requires only QRAM ``read'' operations on the fully classical data (qROM model).
In contrast, the improved tradeoff requires quantum data and also ``write'' QRAM operations.
On the flipside, it shows that fully writable QRAM with quantum data (qRAM) can provide a computational advantage.

\begin{figure}[ht]
\begin{center}
\input{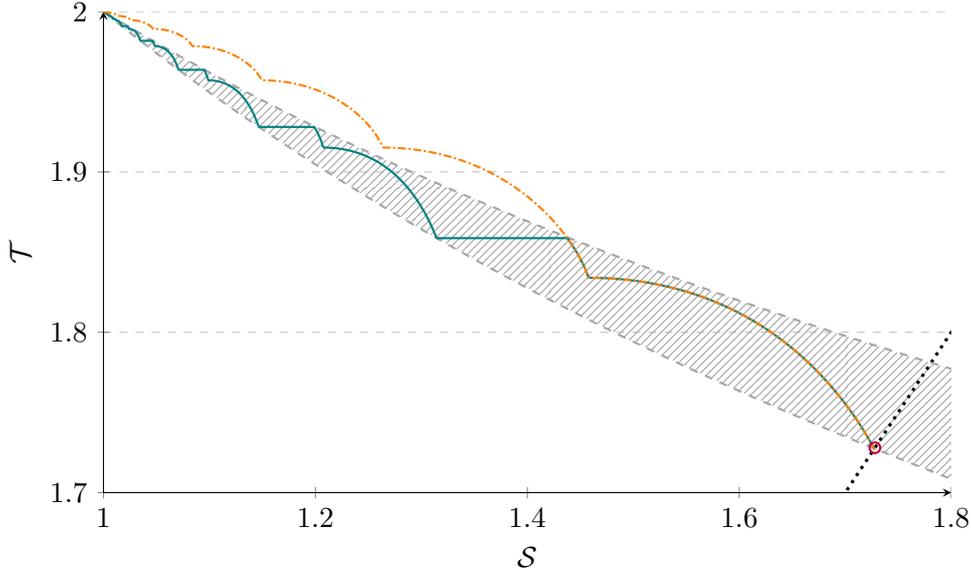}
\end{center}
\caption{Time-space tradeoffs for the divide \& conquer problems.
A point $(\mathcal S, \mathcal T)$ denotes an algorithm with time complexity $\widetilde\Oh(\mathcal T^n)$ that uses at most $\widetilde\Oh(\mathcal S^n)$ QRAM space.
The orange line denotes the optimization tradeoff of Theorem \ref{thm:dnq-opt}.
The teal line denotes the improved tradeoff of Theorem \ref{thm:dnq-imp}.
It is bounded below by $\mathcal T = 2/\mathcal S^{0.268}$ and above by $\mathcal T = 2/\mathcal S^{0.201}$ (dashed gray).
The dotted line shows the function $\mathcal T = \mathcal S$, giving a barrier to improving the tradeoffs after $\mathcal S \approx 1.728$, which corresponds to the quantum dynamic programming algorithm of \cite{ABIKPV19} (red circle).}
\label{fig:dnq-short}
\end{figure}

\paragraph{Permutation Problems.}

The second type of problems can be informally described as finding the minimum of some function $f : S_n \to \mathbb N$, where $S_n$ is the set of permutations on $n$ elements.
These problems reduce to computing values of some function $f(L,M,R) = \min_{\pi \in S_M} \text{cost}(L, \pi, R)$, where $L \cup M \cup R = [n]$ is a partition, $\text{cost}$ is some function that depends on the problem and $S_M$ is the set of permutations of the elements of $M$.
Formally, we require the condition
\begin{equation}
    f(L,M,R) = \min_{\substack{X \cup Y = M \\ |X| = k}} g(f(L,X,Y\cup R), f(L\cup X, Y, R)),
\end{equation}
for some problem-specific function $g$ computable in $\Oh(\poly n)$ time.
The answer then is given by $f(\varnothing,[n],\varnothing)$.
In other words, this condition says that the optimal permutation can be found by decomposing it into subsequences, for which we can find the cost independently and combine them afterwards.

For such problems, Bellman, Held and Karp dynamic programming gives a classical algorithm with $\widetilde\Oh(2^n)$ time and space complexity.
This class includes a multitude of graph vertex ordering problems like \textsc{Treewidth}, \textsc{Pathwidth}, \textsc{Cutwidth}, \textsc{Minimum Fill-In}, \textsc{Optimal Linear Arrangement} etc.~\cite{Bodlaender2012}.

As introduced in \cite{ABIKPV19}, such problems can be modeled by the \textsc{Hypercube Path} query problem: given a subgraph $G$ of the $n$-dimensional Boolean hypercube accessible by queries to the edges of the hypercube, determine whether a path exists between $0^n$ and $1^n$ in $G$.
A quantum algorithm for this problem gives an algorithm with only an $\Oh(\poly n)$ factor increase in the time and space complexity for any of the permutation problems.
The time and space complexity of their quantum algorithm is $\widetilde\Oh(1.817^n)$.

While optimizing the tradeoff of the original divide \& conquer quantum algorithm (Theorem \ref{thm:dnq-opt}) is relatively straightforward, it becomes quite challenging in the case of the \textsc{Hypercube Path}.
The quantum algorithm of \cite{ABIKPV19} is recursive: it calls itself on subcubes of the original hypercube.
If the sole goal is to minimize the time complexity of the algorithm, we need to solve a short system of equations with a handful of parameters, which can be computed numerically without much effort (see Section 3.1 in \cite{ABIKPV19}). 

Since we are interested in the time complexity of the algorithm with limited space, the situation is much more difficult.
By the construction of the algorithm, each recursive level has a constant number of recursive calls and the depth of the recursion is $\Oh(\log n)$.
Therefore, if the allowed space amount is $\widetilde\Oh(\mathcal S^n)$, then each recursive call is still allowed to use $\widetilde\Oh(\mathcal S^n / \poly(n)) = \widetilde\Oh(\mathcal S^n)$ space (since in total there are $c^{\Oh(\log n)} = \Oh(\poly n)$ recursive calls).
That means that while the dimension of the subcubes decreases, the total amount of available space essentially does not decrease.
Thus, each recursive call must have completely different parameters from its parent call, considerably amplifying the amount of independent parameters and making it unfeasible for numerical optimization.

Nonetheless, we find a way to compute this time-space tradeoff in a feasible time by using dynamic programming itself to compute the required time complexities.
The approach is described in Section \ref{sec:hp}, by which we obtain the following result:

\begin{restatable}{theorem}{hpopt} \label{thm:hp-opt}
    For any maximum available amount $\mathcal S^n$ of QRAM space between $\widetilde\Oh(\poly n)$ and $\widetilde \Oh(1.817^n)$, there is a quantum algorithm for the permutation problems with time complexity $\widetilde\Oh(\mathcal T^n)$, where $\mathcal T \in [2/\mathcal S^{0.161},2/\mathcal S^{0.099}]$.
\end{restatable}

Then we examine the classical time-space tradeoff called the \emph{pairwise scheme} introduced in \cite{PK09}.
We combine it with a quantum algorithm for dynamic programming in $n$-dimensional lattice graphs from \cite{GKMV21} to obtain the following time-space tradeoff:

\begin{restatable}{theorem}{hpps} \label{thm:hp-ps}
     For any maximum available amount $\mathcal S^n$ of QRAM space in the range $\mathcal S \in [1.631, 1.826]$, there is a quantum algorithm for the permutation problems with time complexity $\widetilde\Oh(\mathcal T^n)$, where $\mathcal T = 2.511/\mathcal S^{0.527}$.
     By fractalizing this algorithm, we obtain a tradeoff for any space amount between $\widetilde\Oh(\poly n)$ and $\widetilde \Oh(1.826^n)$ with time complexity $\widetilde\Oh(\mathcal T^n)$, where $\mathcal T \in [2/\mathcal S^{0.151},2/\mathcal S^{0.088}]$.
\end{restatable}

The two tradeoffs are depicted in Figure \ref{fig:hp-short}, the first one being more optimal.
Since it is obtained numerically, the graph is not continuous but given by a discrete set of points.
Both tradeoffs follow the fractal-like pattern because of Lemma \ref{thm:fract}.
        
\begin{figure}[ht]
\begin{center}
\input{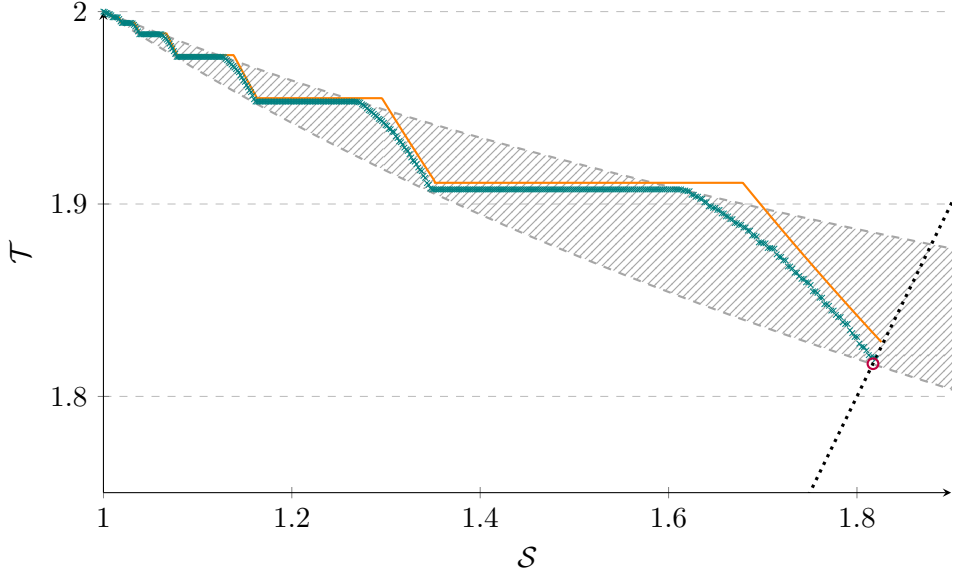}
\end{center}
\caption{Time-space tradeoffs for the permutation problems.
A point $(\mathcal S, \mathcal T)$ denotes an algorithm with time complexity $\widetilde\Oh(\mathcal T^n)$ that uses at most $\widetilde\Oh(\mathcal S^n)$ QRAM space.
The discrete teal plot denotes the optimization tradeoff of Theorem \ref{thm:hp-opt}.
It is bounded below by $\mathcal T = 2/\mathcal S^{0.161}$ and above by $\mathcal T = 2/\mathcal S^{0.099}$ (dashed gray).
The orange line denotes the quantum pairwise scheme tradeoff of Theorem \ref{thm:hp-ps}.
The dotted line shows the function $\mathcal T = \mathcal S$, giving a barrier to improving the tradeoffs after $\mathcal S \approx 1.817$, which corresponds to the quantum \textsc{Hypercube Path} dynamic programming algorithm of \cite{ABIKPV19} (red circle).}
\label{fig:hp-short}
\end{figure}

While the second tradeoff has a strictly larger time complexity, it has a few benefits: first, it is quite close to the first tradeoff.
Secondly, we can write down an explicit formula for its running time and thirdly, the corresponding algorithm is more concise and has a smaller number of parameters.
Both of the tradeoffs use QRAM with ``write'' operations and quantum data (qRAM).

\paragraph{Our Techniques.}
Here we briefly summarize the techniques we use to obtain our tradeoffs.

The divide \& conquer tradeoff of \Cref{thm:dnq-opt} we obtain by optimizing the parameters in the algorithm of \cite{ABIKPV19}.
The improved tradeoff of \Cref{thm:dnq-imp} we obtain by combining the Gurevich and Shelah classical tradeoff \cite{GS87} with the quantum dynamic programming algorithm for \textsc{TSP} \cite{ABIKPV19}: this tradeoff is surprising, as the Gurevich and Shelah tradeoff is already being applied once in the original quantum algorithm.

The optimization tradeoff of \Cref{thm:hp-opt} is obtained by optimizing the (many) parameters of the \textsc{Hypercube Path} quantum dynamic programming algorithm \cite{ABIKPV19}.
This optimization is complicated, and we provide an efficient dynamic programming procedure that allows us to numerically calculate the complexities.
The slightly worse, but cleaner tradeoff of \Cref{thm:hp-ps} we obtain by applying the \emph{quantum dynamic programming on multidimensional lattices} of \cite{GKMV21} to the \emph{pairwise scheme} classical tradeoff \cite{PK09}.
The latter connection is interesting, as it shows that some tradeoffs (the pairwise scheme in this case) may restrict the problem to searching for a path in a subgraph of the hypercube, where the structure of this subgraph still allows for quantum speedups.

\subsection{Related Work} Recently, Jeffery and Pass have demonstrated a quantum time-space tradeoff for directed $st$-connectivity \cite{JP25}.
While the \textsc{Hypercube Path} problem is an instance of directed $st$-connectivity, the time complexity of their algorithm in this case is superexponential, trading off with classical space but always using only $\Oh(n^2)$ quantum space.
However, a simple application of Grover's search over a path vertex in the middle layer of the hypercube and applying the algorithm recursively already gives an $\widetilde\Oh(2^n)$ quantum time algorithm with only $\Oh(\poly n)$ quantum space complexity (which corresponds to $\mathcal S = 1$ and $\mathcal T = 2$ in \Cref{fig:hp-short}).

While our work to the best of our knowledge is the first that examines the quantum time-space tradeoffs for exponential dynamic programming, we've learned that the long-standing classical tradeoffs of Koivisto and Parviainen \cite{KP10} were recently improved, independently and concurrently by \cite{ANW26} and \cite{DK26}.

\section{Preliminaries}

In this section, we introduce preliminary notation and definitions.

\subsection{Notation.}

Throughout the paper, we will use $n$ to denote the size of the problem.
We use the shorthand $\Oh(\poly n)$ to denote a function in $\Oh(n^c)$ for some constant $c$.
We will often use the $\widetilde \Oh(f(n))$ notation to hide $\Oh(\poly \log(f(n)))$ factors.
If an algorithm uses $\widetilde\Oh(\mathcal T^n)$ time and $\widetilde\Oh(\mathcal S^n)$ space, we say that it is an $(\mathcal T, \mathcal S)$ time-space tradeoff. 

\subsection{Model.} Our algorithms work in the standard quantum circuit model, augmented with QRAM (Quantum Random Access Memory).
In QRAM, we will assume that we have a memory of size $N = \widetilde \Oh(\mathcal S^n)$ data items denoted by $\ket{D_1}, \ldots, \ket{D_N}$.
We will assume that each data item is supported on $\Oh(\poly n)$ qubits (or bits, if the data is classical).
This is similar to the classical word RAM model (here we have algorithms of exponential complexity; hence, the size of the word is polylogarithmic of that).

A QRAM ``read-only'' operation implements a unitary mapping
\begin{equation}
    \sum_{i=1}^N \alpha_i \ket{i}_a \ket{b}_t \ket{D_1,\ldots,D_N}_m \longmapsto \sum_{i=1}^N \alpha_i \ket{i}_a \ket{b\oplus D_i}_t \ket{D_1,\ldots,D_N}_m,
\end{equation}
where subscripts $a$, $t$ and $m$ denote \emph{address}, \emph{target} and \emph{memory} registers.
On the other hand, a stronger QRAM ``read-write'' operation implements a unitary mapping
\begin{equation}
    \sum_{i=1}^N \alpha_i \ket{i}_a \ket{b}_t \ket{D_1,\ldots,D_N}_m \longmapsto \sum_{i=1}^N \alpha_i \ket{i}_a \ket{D_i}_t \ket{D_1,\ldots,D_{i-1},b,D_{i+1},\ldots,D_N}_m.
\end{equation}

As we can see, there are four distinct QRAM models an algorithm can employ: classical or quantum data, read-only or read-write operations.
In this paper we use only two: classical data, read-only operations (\textit{qROM}); quantum data, read-write operations (\textit{qRAM}).
The terminology of QRAM is not fully consistent throughout the literature, see e.g.~\cite{JR25,ABDLS24} for more definitions and discussion.

Throughout the paper, we will assume that QRAM operations take unit time.

\subsection{Tools.}

We will frequently use the following well-known approximation:
\begin{theorem}[Entropy approximation] \label{thm:be}
    For any $k \in [0, 1]$, we have
    \begin{equation}
        \binom{n}{\leq k} = \sum_{i=0}^k \binom{n}{k} \leq 2^{\be\left(\frac k n\right)n},
    \end{equation}
    where $\be(x) = -(x \log_2 x + (1-x) \log_2 (1-x))$ is the binary entropy function.
\end{theorem}

The key subroutine of our tradeoffs is Grover's search.
We use its updated version:
\begin{theorem}[Quantum minimum finding with erroneous oracle, Corollary 10 in \cite{ABBLS25}] \label{thm:gro}
    Let $\mathcal A : [N] \to [\Sigma]$ be a bounded-error quantum algorithm with running time $T$ (correct with constant probability).
    Then there exists a bounded-error quantum algorithm that computes $\min_{i \in [N]} \mathcal A(i)$ in $\Oh(\sqrt{N}(\log N + T))$ time.
\end{theorem}

We briefly mention here the benefits of this new variation.
The quantum algorithms of \cite{ABIKPV19} (and consequently, our tradeoffs) call Grover's algorithm recursively, some of them to up to $\Oh(\log n)$ recursive depth.
Since the error probability of inputs is required to be reduced to $\Oh(1/c^n)$ for each of Grover's search applications, this introduces a factor of $n^{\Oh(d)}$ in the overall complexity (if $d$ is the depth of the recursion), which can be as large as $n^{\Oh(\log n)}$.
This doesn't affect the asymptotic complexity of the algorithms, since this factor is subexponential in $n$.
However, it can still be a substantial multiplicative factor.

Note that complexity-wise, $\Oh(c^n f(n)) = \Oh((c+\varepsilon)^n)$ for any $\varepsilon > 0$ if $f(n)$ is subexponential in $c$.
Hence, for example, complexity $\Oh(1.728^n)$ can be slightly deceptive, as the more accurate asymptotics are described by $n^{\Oh(\log n)} 1.727391...^n$.
Theorem \ref{thm:gro} gives an improved complexity and a clearer view: since now no error reduction is necessary, only a single additional $\Oh(n)$ factor comes up in the complexity from Grover's searches.
This lets us write our asymptotic complexities as $\widetilde\Oh(c^n)$, which highlights that the hidden factor is only polynomial in $n$.

\subsection{Fractalization}

Here we describe a useful general property of time-space tradeoffs:
\begin{lemma}[Fractalization] \label{thm:fract}
    For any divide \& conquer or permutation problem, a quantum time-space tradeoff $(\mathcal T, \mathcal S)$ implies a quantum $(\sqrt{2\mathcal T}, \sqrt{\mathcal S})$ tradeoff.
    All of the points $(\mathcal T', \mathcal S')$ iteratively obtained this way lie on the curve $\mathcal T' = 2/\mathcal S'^c$, where $c$ is the number such that $\mathcal T = 2/\mathcal S^c$.
\end{lemma}

\begin{proof}
    If the problem is of the divide \& conquer type, suppose we run a single Grover's search over the sets of size $n/2$ using Equation \ref{eq:dnq} and then apply the given tradeoff for those sets.
    In the case of a permutation problem, consider running Grover's search over all possible sets of size $n/2$ as the first half of the permutation, and then applying the given tradeoff to find the optimal permutation of each half.
    In both cases the time complexity is given by
    \begin{equation}
        \widetilde\Oh\left(\sqrt{\binom{n}{n/2}}\mathcal T^{n/2}\right) = \widetilde\Oh(\sqrt{2\mathcal T}^n)
    \end{equation}
    and the space complexity is only $\widetilde\Oh(\mathcal S^{n/2}) = \widetilde\Oh(\sqrt{\mathcal S}^n)$.

    The second part follows since if $2/\mathcal S^c = \mathcal T$, then $2/ \sqrt{\mathcal S}^c = \sqrt{2\mathcal T}$ is the same identity.
\end{proof}

This explains why the tradeoffs in Figures \ref{fig:dnq-short} and \ref{fig:hp-short} follow a fractal-like pattern.
Moreover, this shows why the coinciding part of the two tradeoffs of Theorems \ref{thm:dnq-opt} and \ref{thm:dnq-imp} (with $\mathcal S \in [1.438, 1.728]$) repeats for smaller values of $\mathcal S$, with horizontal plateaus in between, see Figure \ref{fig:dnq-short}: this happens because the improved tradeoff can be obtained by fractalizing this segment.

Classically, fractalization from $(\mathcal T, \mathcal S)$ implies a $(2\sqrt{\mathcal T}, \sqrt{\mathcal S})$ tradeoff.
Thus, as a side note, one can fractalize the classical time-space tradeoffs of \cite{KP10} to show that one can achieve a time-space product less than $4$ at almost all points $\mathcal S \in [1,2]$ instead of just the segment $\mathcal S \in [1.452,2]$.

\section{Divide \& Conquer Tradeoffs} \label{sec:dnq}

\subsection{Quantum Dynamic Programming}

Our tradeoffs are based on the original algorithm of \cite{ABIKPV19} for the \textsc{Travelling Salesman Problem}, which we retrace here.
Note that for the \textsc{Travelling Salesman Problem} the recurrences below include additional parameters, but the recurrences are essentially the same and the overall complexities stay identical, up to $\Oh(\poly n)$ factors, so we omit these details for the sake of simplicity.

\paragraph{Algorithm Description.} 
The algorithm performs in two parts, the classical precomputation and the quantum search stages.
For the precomputation, we use Equation \ref{eq:dnq} with $k=1$ to get the recurrent relation
\begin{equation}
    f(S) = \min_{x \in S} g(f(S \setminus \{x\}), f(\{x\})).
\end{equation}
Assuming $f(\{x\})$ is trivially computable for all $x \in [n]$, the precomputation computes the value of all $f(S)$ in order of increasing size of $S$ until $|S| \leq \alpha n$, using dynamic programming.
All of these values are then stored in qROM.

In the quantum search stage, we use Equation \ref{eq:dnq} with $k = |S|/2$ to get
\begin{equation}
    f(S) = \min_{\substack{X \subset S \\ |X| = |S|/2}} g(f(S \setminus X), f(X)).
\end{equation}
We start our algorithm with $S = [n]$ and we run Grover's search from Theorem \ref{thm:gro} over the sets $X$.
For each of these sets, we apply Grover's search recursively using the same recurrence.

The recursion continues until $|S| \leq 2\alpha n$, when we apply Grover's search to Equation \ref{eq:dnq} with $k=\alpha n$,
\begin{equation}
    f(S) = \min_{\substack{X \subset S \\ |X| = \alpha n}} g(f(S \setminus X), f(X)).
\end{equation}
Here the values of $f(S\setminus X)$ and $f(X)$ can be obtained from qROM in constant time.

\subsection{Optimization Tradeoff}

Varying the parameter $\alpha$ naturally gives a time-space tradeoff.
By analyzing the time and space complexity, we obtain the following result:

\dnqopt*

\begin{proof}
To compute $f(S)$ for a single $S$ during the precomputation stage, we need to examine all $x \in S$, which are $\Oh(n)$.
The total number of sets examined during the precomputation is equal to $\binom{n}{\leq \alpha n}$.
Thus the time and space complexity at the first stage are equal to
\begin{equation}
    T_P = \widetilde\Oh\left( \binom{n}{\leq \alpha n} \right) = \widetilde\Oh\left( 2^{\be(\alpha)n} \right)
\end{equation}
by Theorem \ref{thm:be}.

In the second stage, since $\alpha \geq 1/2^{k+1}$, then $k$ is the depth of the recursion (i.e.~for the last level, $|S| = n/2^k$).
If $|S| = n/2^i$ for $0 \leq i < k$, then Grover's algorithm searches over sets of size $n/2^{i+1}$.
By Theorems \ref{thm:gro} and \ref{thm:be}, the number of iterations of Grover's search on the $i$-th level then is
\begin{equation}
    \Oh\left(\sqrt{\binom{n/2^i}{n/2^{i+1}}}\right) = \Oh\left(\sqrt{2^{\be(1/2)n/2^i}}\right) = \Oh(2^{n/2^{i+1}}).
\end{equation}
At the last recursive call, the number of options to examine is $\binom{n/2^k}{\alpha n}$.
Thus the number of Grover's iterations is
\begin{equation}
    \Oh\left(\sqrt{\binom{n/2^k}{\alpha n}}\right) = \Oh\left( 2^{\be(\alpha 2^k)n/2^{k+1}}\right).
\end{equation}

The overall time complexity of the search part then is equal to
\begin{equation}
    T_Q = \widetilde\Oh\left( 2^{n \left(\sum_{i=0}^{k-1} \frac{1}{2^{i+1}}+\frac{\be(\alpha 2^k)}{2^{k+1}}\right)}\right) = \widetilde\Oh\left(2^{n\left(1-\frac{2-\be(2^k\alpha)}{2^{k+1}}\right)}\right).
\end{equation}
Note that the $\Oh(\poly(n))$ factor here comes from two places: first, only a single $\Oh(n)$ overhead factor from the topmost Grover's search; secondly, a factor of $\Oh(c^{\log n}) = \Oh(\poly n)$ comes from the $\Oh(\log n)$ recursive applications of Grover's search.

The total space complexity of the search part is only $\Oh(\poly n)$, thus negligible compared to the precomputation (moreover, it doesn't use QRAM).

Altogether, the total space complexity is equal to $\widetilde\Oh(\mathcal S^n) = T_P$ and the time complexity of the algorithm is equal to $\widetilde\Oh(\mathcal T^n) = T_P+T_Q = \Oh(\max(T_P, T_Q))$.
By balancing the two terms, the smallest time complexity is obtained at $\alpha \approx 0.236$ and achieves $\mathcal T = \mathcal S \approx 1.728^n$ (this corresponds to the algorithm of \cite{ABIKPV19}).

For greater $\alpha$, we have $\mathcal S = \mathcal T > 1.728$ which is unfavorable.
For other $\alpha < 0.236$, this gives a time-space tradeoff with
\begin{equation}
    \mathcal S = 2^{\be(\alpha)}, \quad \mathcal T = 2^{1-\frac{2-\be(2^k\alpha)}{2^{k+1}}}.
\end{equation}

Finally, it is also clear the the algorithm requires only qROM memory.
\end{proof}

\subsection{Improved Tradeoff}

The key idea for the improved tradeoff is to use the Gurevich and Shelah approach \cite{GS87}.
Classically, this algorithm runs the divide \& conquer strategy recursively until the size of the set is $s=n/2^k$, at which point the answer is computed using exponential dynamic programming.
The time complexity can be calculated to be $\widetilde\Oh(4^n2^{-s})$ and the space complexity is $\Oh(2^s)$, since the memory is reused for the final subproblems.
Our quantum tradeoff employs the same strategy, except (a) applying Grover's search in the divide \& conquer process and (b) running the quantum algorithm from Theorem \ref{thm:dnq-opt} for the subproblems at the end of the recursion.

\paragraph{Algorithm Description.}

The algorithm is given two parameters, $\alpha \in [0,1/2]$ and $\beta \in [0, 1]$.
Suppose that the algorithm is given a set $S$ and has to compute $f(S)$:
\begin{enumerate}
    \item \label{itm:stp1} If $|S| > 2\beta n$, then it uses Grover's search over sets of size $k=|S|/2$ by Equation \ref{eq:dnq} and calls itself recursively.
    \item \label{itm:stp2} If $\beta n < |S| \leq 2\beta n$, it applies Grover's search over sets of size $k=\beta n$ by Equation \ref{eq:dnq} and calls itself recursively.
    \item \label{itm:stp3} If $|S| \leq \beta n$, then the algorithm calculates $f(S)$ using the algorithm of Theorem \ref{thm:dnq-opt} with parameter $\alpha$.
\end{enumerate}

Note that since the algorithm of Theorem \ref{thm:dnq-opt} is called inside the application of Grover's search, we now need qRAM instead of qROM to be able to write in the same memory at different quantum branches of the computation.

\paragraph{Complexity.}

By analyzing the asymptotic complexity, we obtain the following estimates:
\begin{theorem}
    Suppose that $k \geq 1$, $m \geq 0$ are integers and $\alpha \in [1/2^{k+1},1/2^k]$, $\beta \in [1/2^{m+1},1/2^m]$ real numbers.
    Let $\R(\rho,\ell) := 1-\frac{2-\be(2^\ell\rho)}{2^{\ell+1}}$.
    There exists a quantum algorithm for the divide \& conquer problems with time and space complexities $\widetilde\Oh(\mathcal T^n)$, $\widetilde\Oh(\mathcal S^n)$, where
    \begin{equation}
        \mathcal T = 2^{\R(\beta,m) + \beta\max\{\R(\alpha,k), \be(\alpha)\}}, \quad \mathcal S = 2^{\beta\be(\alpha)}.
    \end{equation}
\end{theorem}

\begin{proof}
    The total time complexity of the Steps \ref{itm:stp1} and \ref{itm:stp2} has already been analyzed in Theorem \ref{thm:dnq-opt}, as in the second stage of that algorithm we run Grover's search with divide \& conquer until sets of size $\alpha n$.
    In this case, it is equal to $\widetilde\Oh\left(2^{n\R(\beta,m)}\right)$.
    The algorithm requires only $\Oh(\poly(n))$ memory because of Grover's searches.

    Step \ref{itm:stp3} requires $\widetilde\Oh(2^{\beta n\max\{\R(\alpha,k), \be(\alpha)\}})$ time and $\widetilde\Oh(2^{n\beta\be(\alpha)})$ space by Theorem \ref{thm:dnq-opt}.
    By multiplying the obtained complexities, the statement follows.
\end{proof}

By numerically analyzing this result for different $\alpha$ and $\beta$, we conclude that competitive tradeoffs are obtained only when $\beta = 1/2^m$ for some $m \geq 0$.
In that case
\begin{equation}
    \R(\beta,m) = 1-\frac{2-\be(2^m\cdot 1/2^m)}{2^{m+1}} = 1-\beta
\end{equation}
and the time complexity simplifies to
\begin{equation}
    \mathcal T = 2^{\max\left\{1-\beta\frac{2-\be(2^k\alpha)}{2^{k+1}},1-\beta + \beta\be(\alpha)\right\}}.
\end{equation}
The Pareto-optimal points give the teal curve in Figure \ref{fig:dnq-short}.

This line can also be obtained by fractalization, which also implies the following result, giving easy-to-compute lower and upper bounds for the tradeoff:
\dnqimp*

\begin{proof}
    By Lemma \ref{thm:fract}, the original $(1.728, 1.728)$ algorithm implies a $(\sqrt{2\cdot 1.728}, \sqrt{1.728}) \approx (1.859, 1.315)$ tradeoff, and this one implies $(1.928, 1.147)$ tradeoff, and so on.
    The described points satisfy some curve $2/\mathcal S^c = \mathcal T$ by Lemma \ref{thm:fract}.
    By solving $2/1.727391^c=1.727391$, we obtain $c \approx 0.268$.

    Similarly we can examine the leftmost point at which the two curves in Figure \ref{fig:dnq-short} coincide.
    We can solve this numerically and find that it happens at $\mathcal S \approx 1.438$, with $\mathcal T \approx 1.859$.
    By fractalization, solving $2/1.438^c = 1.859$ gives a curve $2/\mathcal S^{0.201}$ on which lie all the tradeoff points that can be obtained from $(1.859, 1.438)$ as described in the beginning of the proof.
    We can then verify numerically that this gives an upper bound on the improved time-space tradeoff.
\end{proof}

\section{Permutation Problem Tradeoffs} \label{sec:hp}

In this section, we investigate time-space tradeoffs for permutation problems through two frameworks. The first is the \textsc{Hypercube Path} problem introduced by \cite{ABIKPV19}. The second is the pairwise scheme of \cite{KP10}.  

\subsection{Hypercube Path}\label{ssec:hp}

We revisit the \textsc{Hypercube Path} problem, together with its original quantum algorithm. We analyze its performance under restricted quantum memory, obtaining new upper bounds on the achievable time-space tradeoffs. In addition, we introduce a dynamic programming method that computes the optimal running time for any prescribed space budget, and we show and discuss the results of this tradeoff optimization.

\paragraph{Problem Definition.} In the \textsc{Hypercube Path} problem, we are given query access to a subgraph $G$ of the directed $n$-dimensional hypercube $Q_n$.
The vertices of $Q_n$ are the bit strings in $\{0,1\}^n$ and a directed edge $x \rightarrow y$ exists if $y$ can be obtained from $x$ by flipping a single bit from $0$ to $1$.
The goal is to determine whether there exists a directed path from the vertex $0^n$ to $1^n$.
The graph $G$ is given by an oracle that, given two vertices $x$ and $y$, returns whether there is an edge $x \to y$.

\paragraph{Algorithm Description.} Let $k\in \mathbb N$ be a parameter of the algorithm and $0 < \alpha_1 < \dots < \alpha_k < \alpha_{k+1} = 1/2$ be
constants to be defined later. Let $A_i$ denote the set of vertices of $G$ with Hamming weight $\lfloor \alpha_i n \rfloor$. The algorithm proceeds in three main stages. In the first stage, the algorithm checks if there exists an index $i$ such that $\lfloor \alpha_i n\rfloor =\lfloor \alpha_{i+1} n\rfloor $. This means that the dimension of the hypercube is below some small threshold and the solution is computed classically via dynamic programming.

In the second stage, for each vertex $x$ with Hamming weight at most $\lfloor \alpha_1 n\rfloor$ the algorithm computes whether it is reachable from $0^n$.
It uses dynamic programming and stores this reachability information for all vertices $x \in A_1$ in QRAM.
It then analogously computes whether the vertex $1^n$ is reachable from each vertex in the layer with Hamming weight $n-\lfloor \alpha_1 n\rfloor$.

In the third stage, the algorithm executes Grover's search (see \Cref{thm:gro}) over all vertices $x$ in the middle layer $A_{k+1}$ of the hypercube. For each candidate $x$, it checks if there exists a path from $0^n$ to $x$ (the existence of a path from $x$ to $1^n$ is checked analogously).
The process for verifying this condition proceeds recursively, starting from a target vertex in the highest layer $A_{k+1}$ and decomposing the path step-by-step until it reaches the precomputed base layer $A_1$.

Consider a vertex $x$ in an arbitrary layer $A_i$ and denote by $x \preceq y$ that $x_i \leq y_i$ for all $i \in [n]$.
To determine if there is a path from $0^n$ to $x$, the algorithm must find an intermediate vertex $y$ in the preceding layer $A_{i-1}$ such that: $y\preceq x$, there exists a path from $0^n$ to $y$, and there exists a path from $y$ to $x$ in $G$ in the subcube $\{z \in \{0,1\}^n \mid y \preceq z \preceq x\}$.
This search is implemented quantumly using Grover's search over all potential $y\in A_{i-1}$ with $y\preceq x$.

The recursion unfolds as follows:
\begin{itemize}
    \item \textbf{Base case ($i=1$).} For a vertex $x \in A_1$, the answer is directely retrieved from the classically precomputed lookup table.
    \item \textbf{Recursive step ($i>1$).} For a vertex $x \in A_i$, the algorithm uses Grover's search to find a valid $y \in A_{i-1}$. For each candidate $y$, it makes two recursive calls. The first checks whether $0^n$ is connected to $y$. The second calls the main \textsc{Hypercube Path} algorithm on the subcube between $y$ and $x$.

\end{itemize}

\paragraph{Complexity.} The exponential part of the running time can be expressed as $\widetilde\Oh(\gamma^{n})$.
Let $T_i$ represent the exponential part of the running time for checking the reachability from $0^n$ to a vertex in layer $A_{i}$, with $T_1 = 1$. The recurrence is given by:
\begin{equation}\label{eq:hp-recursive}
T_{i+1} = \sqrt{\binom{\alpha_{i+1}n}{\alpha_i n}} \left( \gamma^{(\alpha_{i+1}-\alpha_i)n} + T_i \right),
\end{equation}
The total (exponential) running time for the entire algorithm is the sum of the classical preprocessing time and the quantum search time:
\begin{equation}\label{eq:hp-balance}
    \gamma^n = \binom{n}{\leq \alpha_1n} + \sqrt{\binom{n}{n/2}}T_{k+1},
\end{equation}
The optimal value of the constant $\gamma$ is determined by balancing the dominant terms in the total running time, which leads to a system of equations. For $k=6$, the solution is $\gamma \approx 2^{0.861483} \approx 1.816905$. Increasing $k$ beyond 6 results in negligible improvement, indicating convergence. The overall time and space complexity is therefore $\widetilde\Oh(1.817^n)$.

\subsection{Optimization Tradeoff}\label{sssec:hp-tradeoff}

The original algorithm for the \textsc{Hypercube Path} problem \cite{ABIKPV19} assumes an unlimited amount of QRAM, capable of storing an arbitrary number of table entries. Our objective is to characterize the time-space tradeoff of the \textsc{Hypercube Path} algorithm, by determining the best achievable running time for any amount $\mathcal S^n$ of QRAM space.

In contrast to the original algorithm, where the parameters $\alpha_i$ can be found by solving a comparatively simple system of equations, our constrained-memory setting requires a far more complex approach.
Suppose that our memory budget is given by $\widetilde\Oh(\mathcal S^n)$.
Since the maximum depth of the recursion is $\Oh(\log n)$ and the number of recursive calls at each level of recursion is constant, the total number of recursive subproblems solved by the algorithm is $\Oh(\poly n)$.
The allowed memory budget at each recursive subproblem is then still $\widetilde\Oh(\mathcal S^n/\poly n) = \widetilde\Oh(\mathcal S^n)$.

Therefore, the dimension of the subcubes decreases in the recursion, while the memory budget stays essentially the same. 
Thus each recursive call operates on a subproblem that may require a distinct set of optimal parameters $\alpha_i$, independent of those in the parent call. 
This recursive dependency forces us to optimize over several parameters at every level of the recursion, rendering a direct numerical approach computationally intractable.
Consequently, solving this generalized time-space tradeoff problem is fundamentally more complex than the parameter optimization for the original algorithm, which assumed unbounded space.

We analyze such a tradeoff for all values of $k=1,\dots,6$, where $k$ denotes the number of layers used to partition the hypercube. For each such $k$ and for any value of $\mathcal S^n$, we determine the time complexity $\mathcal T^n$ of solving the \textsc{Hypercube Path} problem.
To formalize this, we introduce a parameter $c\in[0,1]$ representing the relative dimension of the hypercube: at the beginning of the computation we have $c=1$, and at each recursive step the dimension of the hypercube shrinks, decreasing the value of $c$. We denote by $\mathcal T(c)^n$ the time complexity of solving the problem on a hypercube of relative dimension $c$, under fixed values of $k$ and $\mathcal S$. In particular, $\mathcal T(1)^n=\mathcal T^n$, so that $\mathcal T(c)^n$ will be used when analyzing intermediate subproblems, while $\mathcal T^n$ refers to the overall time complexity of the algorithm. 

We first formalize the memory constraint and the associated objective function to be minimized. We then detail our methodology for computing the optimal value $\mathcal T^n$, for all $\mathcal S^n$ and $k$. 

\paragraph{Single Intermediate Layer Analysis.}
In the case $k=1$, the hypercube is partitioned into two layers: $A_1$ consisting of vertices of Hamming weight $\lfloor \alpha_1 c n \rfloor$, and $A_2$ consisting of vertices of weight $\lfloor \alpha_2 c n \rfloor$ with $\alpha_2 = \frac{1}{2}$ (see \Cref{fig:hp-ill-2-l}).     Let $\lambda_c$ be the parameter determining the maximum allowed size of the DP table computed during preprocessing. This parameter must be chosen to satisfy the condition $\mathcal{S}^{n} = \binom{cn}{\leq \lambda_c cn}$. Subsequently, the parameter $\alpha_1$ is chosen to satisfy the inequality $\alpha_1 \leq \lambda_c$.
The time complexity $\mathcal{T}^n$ is computed as follows.

\textbf{Base case ($\alpha_1 = \frac 1 2$).} The entire dynamic programming table fits into the memory, so the solution can be directly computed via the classical algorithm, giving:
\begin{equation}\label{eq:hp-k-1-bs}
\mathcal T(c)^n = 2^{cn}.
\end{equation}

\textbf{Recursive case ($\alpha_1 < \frac 1 2$).} In this case, the available memory $\mathcal S^{cn}$ is insufficient for storing the whole dynamic programming table, necessitating a recursive approach. While the algorithm can precompute up to $\mathcal S^{cn}$ entries, the optimal strategy may not use the entire memory budget. The optimal time complexity is therefore given by minimizing over all possible choices of $\alpha_1\in (0,\lambda_c]$ as follows.

\begin{equation}\label{eq:hp-k-1-rc}
\mathcal T(c)^n =\min_{\alpha_1\in (0,\lambda_c]}\max \begin{cases}
\binom{cn}{\leq\alpha_1 cn}, \\
\sqrt{\binom{cn}{\alpha_2 cn}} \cdot \sqrt{\binom{\alpha_2cn}{\alpha_1 cn}} \cdot \mathcal T\left((\alpha_2 - \alpha_1)  c\right)^n.
\end{cases} 
\end{equation}

\noindent The first term in the $\max$ function represents the classical preprocessing time. The second term describes the quantum recursive time, which is a product of three components: (i) $\sqrt{\binom{cn}{\alpha_2 cn}}$ for Grover's search over the middle layer $A_2$; (ii) $\sqrt{\binom{\alpha_2cn}{\alpha_1 cn}}$ for Grover's search over the predecessors in layer $A_1$; and  (iii) $\mathcal T\left((\alpha_2 - \alpha_1)  c\right)^n$ for the recursive call on a subcube of relative dimension $(\alpha_2-\alpha_1)c$.
For a fixed $\alpha_1$, the time complexity is determined by the slower of the two phases: preprocessing or recursion. The overall optimal time $\mathcal T(c)^n$ is found by minimizing this value over all possible $\alpha_1$.

The recurrence for the recursive cases can be expressed asymptotically using the binary entropy function to approximate the binomial coefficients, as stated in \Cref{thm:be}. This entropy-based formulation will be used for all subsequent cases $k = 2, \dots, 6$. 
Moreover, since the time and space complexities grow exponentially with $n$, we analyze their base-$2$ exponents for simplicity. Defining the time exponent as $\tau (c) =\log_2 \mathcal T(c)$, and noting its independence from the absolute input dimension, we omit explicit dependence on $n$ hereafter.
The recurrence for $k=1$ can be rewritten as:

\begin{equation}\label{eq:hp-k-1-rc-e}
\tau(c) =\min_{\alpha_1\in (0,\lambda_c]}\max \begin{cases}
\be(\alpha_1) c, \\
\frac{1}{2}(\be(\alpha_2) c + \be(\frac{\alpha_1}{\alpha_2}) \alpha_2c) + \tau\left((\alpha_2 - \alpha_1)  c\right).
\end{cases} 
\end{equation}

\paragraph{Two Intermediate Layer Analysis.}
We now extend our analysis to the case $k=2$, before generalizing to arbitrary $k$. In this case, the hypercube is partitioned into three layers: $A_1$, $A_2$, and $A_3$, containing vertices of Hamming weight $\lfloor \alpha_1 cn\rfloor$, $\lfloor \alpha_2cn\rfloor$, and $\lfloor \alpha_3 cn\rfloor$, respectively, where by definition $\alpha_3=\frac{1}{2}$ (see \Cref{fig:hp-ill-3-l}). The layer $A_1$ determines the size of the DP table computed during the classical preprocessing step. Layer $A_3$ is the middle layer, and $A_2$ is an intermediate layer between $A_1$ and $A_3$. The time complexity $\mathcal T(c)$ is computed as follows.

\textbf{Base case ($\alpha_1 = \frac 1 2$).} $\tau(c) = c$ similarly as before.

\textbf{Recursive case ($\alpha_1 < \frac 1 2$).} The available memory is insufficient for a direct solution, necessitating a recursive approach. The optimal strategy may not use the entire memory budget, and the optimal value for the intermediate parameter $\alpha_2$ is not known a priori, as it can range from $\alpha_1$ to $\alpha_3$. The optimal time complexity is therefore given by minimizing over all possible choices of $\alpha_1 \in (0, \lambda_c]$ and $\alpha_2 \in (\alpha_1, \alpha_3)$. For succinctness, we denote this double minimization by $\displaystyle \min_{\alpha_1, \alpha_2}=\min_{\alpha_1\in(0,\lambda_c]}\min_{\alpha_2\in(\alpha_1,\alpha_3)}$.

\begin{equation}\label{eq:hp-k-2-rc-e}
\tau(c) = \min_{\alpha_1,\alpha_2}\max \begin{cases}
\be(\alpha_1) c, \\
\frac{1}{2}(\be(\alpha_3) c + \be(\frac{\alpha_2}{\alpha_3}) \alpha_3c) + \tau\left((\alpha_3 - \alpha_2)  c\right), \\
\frac{1}{2}(\be(\alpha_3) c + \be(\frac{\alpha_2}{\alpha_3}) \alpha_3c + \be(\frac{\alpha_1}{\alpha_2}) \alpha_2 c) + \tau\left((\alpha_2 - \alpha_1)  c\right).
\end{cases}
\end{equation}

\noindent The first term in the $\max$ function represents the classical preprocessing time. The second term describes the time complexity for the recursive procedure that checks if there exists a path from a vertex in $A_2$ to a vertex in $A_3$. It comprises: Grover's search over $A_3$ ($\frac{1}{2}\be(\alpha_3) c$); Grover's search over predecessors in $A_2$ ($\frac{1}{2}\be(\frac{\alpha_2}{\alpha_3}) \alpha_3 c$); and the recursive call on the subcube between $A_2$ and $A_3$ ($\tau((\alpha_3 - \alpha_2) c)$). The third term corresponds to time complexity for checking if there exists a path from a vertex in $A_1$ to a vertex in $A_2$. It includes the previous two Grover's terms, plus the Grover's search over predecessors in $A_1$ ($\frac{1}{2}\be(\frac{\alpha_1}{\alpha_2}) \alpha_2 c$) and the recursive call on the subcube between $A_1$ and $A_2$ ($\tau((\alpha_2 - \alpha_1) c)$).

\begin{figure}[tb!]
%\centering
\begin{subfigure}[t]{0.475\textwidth}
\includegraphics[page=1,width=0.9\textwidth]{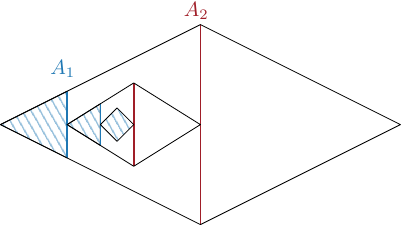}
\centering 
\caption{A $2$-layer decomposition of the hypercube.}
\label{fig:hp-ill-2-l}
    \end{subfigure}
\begin{subfigure}[t]{0.475\textwidth}
        \centering
\includegraphics[page=2,width=0.9\textwidth]{hypercube.pdf} 
\centering 
        \caption{A $3$-layer decomposition of the hypercube.}
\label{fig:hp-ill-3-l}
    \end{subfigure}
\caption{Illustration of the recursive hypercube layering strategy. \Cref{fig:hp-ill-2-l} depicts the case for $k=2$ layers. The initial preprocessing layer $A_1$ is highlighted in blue. Its size, determined by parameter $\alpha_1$, is constrained by the available memory, represented by the blue striped region. The middle layer $A_2$ is highlighted in red. The diagram shows three levels of recursion. At the final level, the subcube is small enough that its entire solution fits within the memory budget and can be computed classically. \Cref{fig:hp-ill-3-l} illustrates the analogous structure for $k=3$ layers. The intermediate layer $A_2$ is highlighted in green.}
\label{fig:hp-illustration}
\end{figure}

\paragraph{General Case Analysis.}
We now generalize our analysis to an arbitrary number of layers $k$. In the $k$-layer configuration, we have layers $A_1, A_2, \dots, A_{k+1}$, where the layer $A_1$ determines the size of the DP table computed in the preprocessing step, $A_{k+1}$ is the middle layer (with $\alpha_{k+1}=1/2$), and all others are intermediate layers. The time complexity $\mathcal T(c)$ is computed as follows.

\textbf{Base case ($\alpha_1 = \frac 1 2$).} $\tau(c) = c$ similarly as before.

\textbf{Recursive case ($\alpha_1 < \frac 1 2$).} The available memory is insufficient for a direct solution, necessitating a recursive approach. The optimal strategy may not use the entire memory budget, and the optimal values for the parameters $\alpha_1,\alpha_2, \dots, \alpha_k$ are not known a priori. The optimal time complexity is therefore given by minimizing over all valid sequences of parameters $\alpha_1, \dots, \alpha_k$ satisfying the chain $0 < \alpha_1 < \alpha_2 < \dots < \alpha_k < \alpha_{k+1}$. For succinctness, we denote this multiple minimization by $$ \min_{\alpha_1,\dots, \alpha_k}=\min_{\alpha_1\in(0,\lambda_c]}\min_{\alpha_2\in(\alpha_1,\alpha_{k+1})}\dots \min_{\alpha_k\in(\alpha_{k-1},\alpha_{k+1})}.$$
The recurrence relation is given by:
\begin{equation}\label{eq:hp-k-gen-rc-e}
\resizebox{0.91\textwidth}{!}{%
$
\tau(c) = \min_{\alpha_1,\dots,\alpha_k}\max \begin{cases}
\be(\alpha_1) c, \\
\frac{1}{2}\Big(\be(\alpha_{k+1}) c+\sum_{j=i}^{k} \be\Big(\frac{\alpha_{j}}{\alpha_{j+1}}\Big) \alpha_{j+1}c\Big)+
\tau((\alpha_{i+1}-\alpha_i) c), & \text{for all } i\in[1,k].\\
\end{cases}
$
}
\end{equation}
\noindent The terms in the $\max$ function generalize the structure from the $k=1$ and $k=2$ cases. The first term represents the exponent of the classical preprocessing time. The inner maximization over $i$ selects the most costly recursive path, where each path corresponds to recursing at a different layer $i$. For a fixed path $i$, the terms represent:  Grover's search over the middle layer $A_{k+1}$; the sum of Grover's search exponents over the predecessors in each subsequent layer down to $A_i$; and the recursive call on the subcube between layers $A_i$ and $A_{i+1}$.

Our analysis of the time and space complexity of the \textsc{Hypercube Path} algorithm for $k = 1, \dots, 6$ yields a family of tradeoff curves $\mathcal T$, parameterized by the memory usage $\mathcal S $. These curves quantify the optimal achievable time complexity under a constrained quantum memory budget. The results of this computation are illustrated in \Cref{fig:hp-levels}.

\begin{figure}[tb!]
\begin{center}
\input{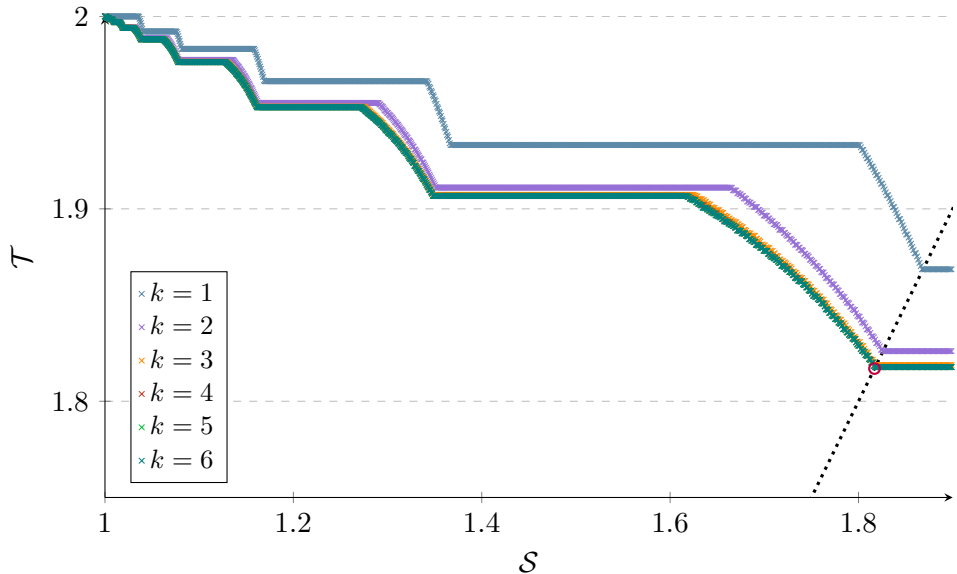}
\end{center}
\caption{Quantum time-space tradeoffs for the \textsc{Hypercube Path} problem. A point at coordinates $(\mathcal{S}, \mathcal{T})$ denotes an algorithm with time complexity $\widetilde{\Oh}(\mathcal{T}^n)$ and QRAM space complexity $\widetilde{\Oh}(\mathcal{S}^n)$.
Each colored curve corresponds to a different recursive strategy parameterized by $k$, the number of layers split per recursive call. Higher values of $k$ yield better tradeoffs.
The dotted line shows the function $\mathcal T = \mathcal S$, giving a barrier to improving the tradeoffs after $\mathcal S \approx 1.817$, which corresponds to the quantum \textsc{Hypercube Path} dynamic programming algorithm of \cite{ABIKPV19} (red circle).}
\label{fig:hp-levels}
\end{figure}

Consistent with the original algorithm, the curves for $k \geq 6$ tend to converge, as further increasing the number of layers does not yield a significant improvement in time complexity.
\Cref{tb:complexities} presents the resulting optimal time complexities for a representative sample of $\mathcal S$ values across all $k=1,\dots,6$.
Note that there is a small discrepancy between the constant $1.816905$... obtained in \cite{ABIKPV19} and our calculations in which we obtain $1.817776$...: this is due to the accumulation of small imprecisions in our calculations.

\begin{table}[tb!]
\centering
\begin{tabular}{ccccccc}

\multicolumn{1}{c||}{} & \multicolumn{1}{c|}{$k=1$} & \multicolumn{1}{c|}{$k=2$} & \multicolumn{1}{c|}{$k=3$} & \multicolumn{1}{c|}{$k=4$} & \multicolumn{1}{c|}{$k=5$} & $k=6$ \\ \hline \hline
\multicolumn{1}{c||}{$\mathcal S =1.0$}              & \multicolumn{1}{c|}{$2$} & \multicolumn{1}{c|}{$2$} & \multicolumn{1}{c|}{$2$} & \multicolumn{1}{c|}{$2$} & \multicolumn{1}{c|}{$2$} & $2$ \\ \hline
\multicolumn{1}{c||}{$\mathcal S =1.2$}              & \multicolumn{1}{c|}{$1.966319$} & \multicolumn{1}{c|}{$1.955016$} & \multicolumn{1}{c|}{$1.953075$} & \multicolumn{1}{c|}{$1.952799$} & \multicolumn{1}{c|}{$1.952799$} & $1.952799$ \\ \hline
\multicolumn{1}{c||}{$\mathcal S =1.4$}              & \multicolumn{1}{c|}{$1.933180$} & \multicolumn{1}{c|}{$1.911044$} & \multicolumn{1}{c|}{$1.907250$} & \multicolumn{1}{c|}{$1.906712$} & \multicolumn{1}{c|}{$1.906712$} & $1.906712$ \\ \hline
\multicolumn{1}{c||}{$\mathcal S =1.6$}              & \multicolumn{1}{c|}{$1.933180$} & \multicolumn{1}{c|}{$1.911044$} & \multicolumn{1}{c|}{$1.907250$} & \multicolumn{1}{c|}{$1.906712$} & \multicolumn{1}{c|}{$1.906712$} & $1.906712$ \\ \hline
\multicolumn{1}{c||}{$\mathcal S =1.8$}              & \multicolumn{1}{c|}{$1.931984$} & \multicolumn{1}{c|}{$1.843999$} & \multicolumn{1}{c|}{$1.830741$} & \multicolumn{1}{c|}{$1.828918$} & \multicolumn{1}{c|}{$1.828918$} & $1.828918$ \\ \hline
\multicolumn{1}{c||}{$\mathcal S = \mathcal T$}              & \multicolumn{1}{c|}{$1.868583$} & \multicolumn{1}{c|}{$1.826044$} & \multicolumn{1}{c|}{$1.818802$} & \multicolumn{1}{c|}{$1.817776$} & \multicolumn{1}{c|}{$1.817776$} & $1.817776$ \\
\end{tabular}
\caption{Representative optimal time complexities illustrating the time-space tradeoff for selected values of the memory parameter $\mathcal S$ across different numbers of layers $k=1,\dots,6$.
The last line shows the smallest time complexity for the quantum \textsc{Hypercube Path} algorithm with $k$ layers.}
\label{tb:complexities}
\end{table}

This approach leads to the following result.

\hpopt*

\begin{proof}
    We can see that fractalization is built-in in the \textsc{Hypercube Path} tradeoff in the following way.
    It conforms to the described approach as for the top recursive level we can take $\alpha_1 = \ldots = \alpha_k = 0$ and $\alpha_{k+1} = 1/2$ (the algorithm as stated, formally requires $\alpha_i$ values to be distinct, but we can take them to be arbitrarily close to each other).
    Thus if $(\mathcal T_k, \mathcal T_k)$ is the time-optimal algorithm for a fixed $k$, then all the tradeoffs obtained by fractalization will be present in our analysis; indeed, these correspond to the leftmost points of all horizontal plateaus in Figure \ref{fig:hp-levels}.
    By Lemma \ref{thm:fract}, for each $k$ these points will lie on a curve $2/\mathcal S^{c_k} = \mathcal T$ for some constant $c_k$.
    For $k = 6$, we can solve $2/1.816905^{c_k} = 1.816905$ and get $c_k \approx 0.161$.
    Hence, the time-space tradeoff is lower bounded by $2/\mathcal S^{0.161}$.
    Similarly we can also find that it is upper bounded by $2/\mathcal S^{0.099}$.
\end{proof}

\subsubsection{Computing the Optimal Values.}\label{sssec:hp-computation}

A central step in our analysis is the computation of $\mathcal T$ for all $\mathcal S\in[1,2]$ and $k=1,\dots,6$. The parameter $\mathcal S$ takes values in a continuous interval, necessitating a discrete approximation with a finite granularity. A naive approach that employs $k$ nested loops over the discrete values of $\mathcal S$ is computationally infeasible, as the precision of the results depends critically on a high-resolution discretization (requiring a granularity of at least $1000$ points for a good approximation). This would impose a prohibitive computational burden even on modern parallel architectures. To overcome this limitation, we employ a dynamic programming strategy that stores intermediate results to avoid redundant calculations. Our implementation for populating the dynamic programming table is provided in \cite{VC25}.

This approach revolves around two interacting DP tables: the \emph{optimal complexity} table $\mathtt{T}$ and the \emph{partial complexity} table $\mathtt{P}$. Since both the time complexity $\mathcal{T}^n$ and space complexity $\mathcal{S}^n$ are exponential in $n$, we simplify our analysis by focusing on their exponents. The entries in our dynamic programming tables $\mathtt{T}$ and $\mathtt{P}$ therefore store the base-$2$ logarithm of the complexity; that is, an entry contains the value $\tau(c) = \log_2 \mathcal T(c)$.
Specifically, an entry $\mathtt{T}[r,s]$ contains the optimal time complexity for determining whether a path exists in the hypercube using a recursion depth of at most $r$ and a memory budget parameterized by $s$, where $\mathcal S^n =2^{sn}$. Meanwhile, an entry $\mathtt{P}[r,s,\alpha_1,i,\alpha_i]$ contains the partial complexity for determining the existence of a path from $0^n$ to a vertex in layer $A_i$ (with parameter $\alpha_i$), a preprocessing layer $A_1$ with parameter $\alpha_1$, a recursion depth $r$, and a memory parameter $s$. The values of $\mathtt{P}$ are defined only for $r\geq 1$ and $i\geq 2$; the case $i=1$ is trivial ($\mathtt{P}[\cdot, \cdot, \cdot, 1, \cdot] = 0$) as the reachability of vertices in $A_1$ is known from the precomputation step.

We now formalize the base and recursive cases for filling these tables. For a fixed number of layer $k$, we distinguish the $(k+1)$-th layer by setting $\alpha_{k+1}=\frac{1}{2}$, as this layer always corresponds to the middle of the hypercube where the final Grover's search is performed.

\bigskip
\textbf{Base case $\mathtt{T}$ ($r=0$).} In the base case with no recursion that uses layers $A_i$, for all values of $s$, the algorithm resorts to a direct quantum search, yielding:
\begin{equation}\label{eq:bc-T}
    \mathtt{T}[0,s] = 1.
\end{equation}

\noindent This corresponds to the exponent of the $\widetilde\Oh(2^n)$-time, $\widetilde\Oh(1)$-space quantum algorithm.

\bigskip
\textbf{Recursive case $\mathtt{T}$ ($r>0$).} The optimal time complexity is obtained by minimizing over all valid choices of the first layer parameter $\alpha_1$, constrained by the available memory $s=\be(\lambda_c)$.

\begin{equation}
    \mathtt{T}[r,s] = \min_{\alpha_1 \in (0, \lambda_c]} \left\{ \max \left\{ \be(\alpha_1), \frac{1}{2} + \mathtt{P}[r, s, \alpha_1, k+1, 1/2] \right\} \right\}.
\end{equation}

\noindent This recurrence mirrors the structure of \Cref{eq:hp-k-gen-rc-e}, balancing the cost of classical preprocessing against the quantum recursive cost of verifying paths from $A_1$ to the middle layer.

\bigskip
\textbf{Base case $\mathtt{P}$ ($i=2)$.} The complexity of checking if there is a path to layer $A_2$ is computed as follows.
\begin{equation}\label{eq:bc-P-2}
\mathtt{P}[r, s, \alpha_1, 2, \alpha_2] = \frac{1}{2} \be\left( \frac{\alpha_1}{\alpha_2} \right) + (\alpha_2 - \alpha_1) \mathtt{T}\left[ r-1, \min\left\{ \frac{s}{\alpha_2 - \alpha_1}, 1 \right\} \right].
\end{equation}
\noindent The first term represents the exponent for Grover's search over layer $A_1$ to find a predecessor of a vertex in $A_2$. The second term is the recursive cost of verifying the path within the subcube of relative dimension $(\alpha_2 - \alpha_1)$ between the two layers.

\bigskip
\textbf{Recursive case $\mathtt{P}$ ($i\geq 3)$.} For general $i$, the cost of reaching layer $A_i$ from $A_1$ is computed by minimizing over all possible intermediate parameters $\alpha_{i-1} \in (\alpha_1, \alpha_i)$ of the preceding layer.
\begin{equation}\label{eq:rc-P-i}
 \renewcommand{\arraystretch}{1.5}
            \resizebox{\textwidth}{!}{$%
\begin{aligned}
\mathtt{P}[r, s, \alpha_1, i, \alpha_i] = \min_{\alpha_{i-1} \in (\alpha_1, \alpha_i)} \Bigg\{
& \frac{1}{2} \be\left( \frac{\alpha_{i-1}}{\alpha_i} \right) + \\
& \max \left\{ \mathtt{P}[r, s, \alpha_1, i-1, \alpha_{i-1}],\; (\alpha_i - \alpha_{i-1}) \mathtt{T}\left[r-1, \min\left\{ \frac{s}{\alpha_i - \alpha_{i-1}}, 1 \right\} \right] \right\}
\Bigg\}.
\end{aligned}$%
}
\end{equation}

\noindent The first term in the sum is the cost of Grover's search for a predecessor in $A_{i-1}$. The $\max$ operator ensures the total cost is dominated by the slower of two processes: the accumulated cost of reaching layer $A_{i-1}$ itself and the recursive cost of exploring the subcube between $A_{i-1}$ and $A_i$.

\bigskip
In summary, this recursive formulation of $\mathtt{T}$ and $\mathtt{P}$ provides a computationally tractable framework for exploring the time-space tradeoff. It avoids the infeasibility of a brute-force discretization over $s$  while enabling precise exploration of the time-memory tradeoff in the recursive quantum algorithm.

\subsection{Pairwise Scheme}\label{ssec:pws}

The pairwise scheme, originally introduced by \cite{KP10} in the classical setting, provides a general framework for solving permutation problems.
Here we show how we can improve it quantumly.

\paragraph{Classical Algorithm.}

First we describe how the pairwise scheme algorithm works to solve the \textsc{Hypercube Path} problem classically.
A valid path from $0^n$ to $1^n$ can be described by a permutation $\pi = (\pi_1, \ldots, \pi_n)$, meaning that the $i$-th edge connects two bit strings that differ in the $\pi_i$-th coordinate.
Let $k \in [0,n/2]$ be an integer parameter and for each pair of integers $\{2i-1,2i\}$, where $i \in [k]$, fix their relative order in this permutation.
Each of the $X = 2^k$ resulting partial orders defines a single subproblem, which  can be evaluated via dynamic programming as follows.

For a set $S \subseteq [n]$, let $1_S$ denote an $n$-bit string where $x_i = 1$ iff $i \in S$.
Then let $f(S) = 1$ iff $1_S$ is reachable from $0^n$ in $G$ and $0$ otherwise.
If a set $S \subseteq [n]$ satisfies the current partial order, we can compute $f(S)$ using
\begin{equation}
    f(S) = \lor_{i \in S} (f(S \setminus i) \land (1_{S \setminus i}, 1_S) \in G).
\end{equation}
If $Y$ is the total number of sets $S$ satisfying the partial order, then using dynamic programming we can compute $f(S)$ for all such $S$ in $\widetilde\Oh(Y)$ time and space (the sets can be listed in complexity $\widetilde\Oh(Y)$ using, for example, breadth-first search).

To calculate $Y$, note that for each pair $\{2i-1,2i\}$, only three of its subsets conform to the partial order: if the order of two elements $a$ and $b$ is fixed to $a \to b$, then the valid subsets are $\varnothing$, $\{a\}$ and $\{a, b\}$ ($\{b\}$ does not conform to the partial order, as $b$ cannot precede $a$).
For all other $n-2k$ elements, they can either appear in $S$ or not without restrictions.
Therefore, $Y = 3^k \cdot 2^{n-2k}$.

This yields an overall complexity of $\widetilde\Oh(XY) = \widetilde\Oh((3/2)^k2^n)$ time and $\widetilde\Oh(Y) = \widetilde\Oh((3/4)^k2^n)$ space, for any $k \in [n/2]$, producing a practical scheme that operates within exponential space $\widetilde\Oh(\mathcal S^n)$, where $\sqrt{3}\leq \mathcal S \leq 2$.  

\paragraph{Quantum Algorithm.}

We combine the pairwise scheme with a quantum algorithm for dynamic programming in $n$-dimensional lattice graphs from \cite{GKMV21} to obtain the following quantum time-space tradeoff:

\hpps*

\begin{proof}
    We will speed up two parts of the classical algorithm quantumly.
    First of all, we can run Grover's search over the subproblems, which requires only $\Oh(\sqrt{X}) = \Oh(\sqrt{2^k})$ Grover's iterations.

    Secondly, we will reduce the dynamic programming in a single subproblem to detecting a path in a \textit{multi-dimensional lattice}.
    Indeed, examine a graph on the set of vertices
    \begin{equation}
        V = \{0,1,2\}^{k} \times \{0, 1\}^{n-2k},
    \end{equation}
    in which the first $k$ coordinates correspond to the $k$ pairs with the fixed relative order and the other $n-2k$ coordinates correspond to elements $[2k+1,n]$.
    Denote the \textit{weight} of a vertex $x \in V$ by $|x| = \sum_{i=0}^{n-k} x_i$.
    Let a directed edge go from $x$ to $y$ in this graph iff there is a single $i$ such that $y_i = x_i+1$ and $y_i = x_i$ for all other $i$.

    Since we are interested in $f([n])$, then the classical dynamic programming over subsets can be reduced to detecting a path from $0^{n-k}$ to $2^k 1^{n-2k}$.
    Indeed, for the first $k$ coordinates, the values $0$, $1$ and $2$ denote the subsets $\varnothing$, $\{a\}$ and $\{a, b\}$.
    The other coordinates have values $0$ and $1$ if the corresponding element is or is not contained in the set.

    In \cite{GKMV21}, the authors proposed quantum algorithms for detecting such a path in a hyperlattice.
    In our case, we have a hyperlattice with $k$ dimensions of size $3$ and $n-2k$ dimensions of size $2$.
    For such a graph, their algorithm achieves time and space complexity
    \begin{equation}
        \widetilde\Oh(2.65907^k \cdot 1.82653^{n-2k}) = \widetilde\Oh(1.82653^n / 1.25465^k) = \widetilde\Oh(\mathcal S^n),
    \end{equation}
    see \texttt{solverD2K5.nb} in \cite{GKMV21data}.

    Our algorithm then requires the same amount of space, and its time complexity is
    \begin{equation}
        \widetilde\Oh(\sqrt{2^k} \cdot 1.82653^n / 1.25465^k) = \widetilde\Oh(1.82653^n \cdot 1.12718^k) = \widetilde\Oh(\mathcal T^n).
    \end{equation}

    We can calculate that the smallest amount of memory that this tradeoff can be applied to is $1.82653^n / 1.25465^{n/2} \approx 1.631^n$.
    When expressing $\mathcal T$ with $\mathcal S$, we obtain that $\mathcal T = 2.511/\mathcal S^{0.527}$.

    Finally, we can extend this tradeoff for all values of $\mathcal S$ using Lemma \ref{thm:fract}.
    We can calculate that the obtained tradeoff is lower and upper bounded by the functions $2/\mathcal S^{0.151}$ and $2/\mathcal S^{0.088}$, respectively.
\end{proof}

The graph of this tradeoff is shown in Figure \ref{fig:hp-short}.
While slightly less efficient than Theorem \ref{thm:hp-opt}, in this case, the algorithm is conceptually simpler and it also comes with an explicit upper bound on its time and space asymptotic complexity.

\section{Acknowledgements}

We thank an unrecognized audience member (but hopefully someday credited) at Joint Estonian-Latvian Theory Days 2018 for asking whether the results \cite{KP10} have relevance to the algorithms of \cite{ABIKPV19}, which started this research project.

J.V.~is supported by the 1.1.1.9 Research application No 1.1.1.9/LZP/2/25/206 of the Activity ``Post-doctoral Research'' ``Practical applications and limitations of quantum search''.

\printbibliography

\end{document}